\input ./style/arxiv-vmsta.cfg
\documentclass[numbers,compress,v1.0.1,seceqn]{vmsta}
\usepackage{vtexbibtags}
\usepackage{mathbh}

\volume{7}% Updated by VTEXPTS2LaTeX.exe, 23.06.2020 10:59
\issue{2}% Updated by VTEXPTS2LaTeX.exe, 23.06.2020 10:59
\pubyear{2020}
\firstpage{113}% Updated by VTEXPTS2LaTeX.exe, 23.06.2020 10:59
\lastpage{134}% Updated by VTEXPTS2LaTeX.exe, 23.06.2020 10:59
\aid{VMSTA152}% Updated by VTEXPTS2LaTeX.exe, 08.04.2020 13:21
\doi{10.15559/20-VMSTA152}% Updated by VTEXPTS2LaTeX.exe, 08.04.2020 13:21
\articletype{research-article}

\startlocaldefs

\allowdisplaybreaks
\numberwithin{equation}{section}

\newtheorem{theorem}{Theorem}[section]
\newtheorem{Lemma}[theorem]{Lemma}
\newtheorem{Corollary}[theorem]{Corollary}
\newtheorem{Proposition}[theorem]{Proposition}
\theoremstyle{definition}
\newtheorem{Remark}[theorem]{Remark}

\endlocaldefs

\begin{document}

\begin{frontmatter}
\pretitle{Research Article}

\title{A pure-jump mean-reverting short rate model}

\author{\inits{M.}\fnms{Markus}~\snm{Hess}\ead[label=e1]{Markus-Hess@gmx.net}}
%\author[]{\inits{}\fnms{}~\snm{}\thanksref{f1}\thanksref{cor1}\ead[label=e2]{}}
%\thankstext[type=corresp,id=cor1]{Corresponding author.}
\address{Independent}
%\address[]{\institution{}, ..., \cny{}}

%\thankstext[id=f1]{}

%\dedicated{}

%\markboth{Authors}{Title}
\markboth{M. Hess}{A pure-jump mean-reverting short rate model}

\begin{abstract}
%We present
A new multi-factor short rate model is presented which is
bounded from below by a real-valued function of time. The mean-reverting
short rate process is modeled by a sum of pure-jump Ornstein--Uhlenbeck
processes such that the related bond prices possess affine representations.
%We also provide
Also the dynamics of the associated instantaneous forward rate is provided
and %derive
a condition is derived under which the model can be market-consistently
calibrated. The analytical tractability of
%our
this model is illustrated by the
derivation of an explicit plain vanilla option price formula. With view on
practical applications,
%we propose
suitable probability distributions are proposed for
the driving jump processes. %We conclude
The paper is concluded by presenting a
post-crisis extension of
%our
the proposed short and forward rate model.
\end{abstract}


\received{\sday{23} \smonth{4} \syear{2019}}% Updated by VTEXPTS2LaTeX.exe, 08.04.2020 13:21
\revised{\sday{27} \smonth{3} \syear{2020}}% Updated by VTEXPTS2LaTeX.exe, 08.04.2020 13:21
\accepted{\sday{27} \smonth{3} \syear{2020}}% Updated by VTEXPTS2LaTeX.exe, 08.04.2020 13:21
\publishedonline{\sday{20} \smonth{4} \syear{2020}}

\end{frontmatter}

% Konvertuota OpenXml2Latex, v1.36.1 / 93e2876c355a
%
%
%

%s1 #&#
\section{Introduction}\label{sec1}

Stochastic interest rate\index{interest rate} models play an important role in the modeling of
financial markets. The literature essentially distinguishes between short
rate models, forward rate\index{forward rate} models and market models.\index{market models} In the sequel, we give
a brief survey on the different classes of term structure\index{term structure} models. For more
detailed information, the reader is referred to the respective research
articles or the textbooks \cite{7,21} and \cite{26}.

Widely applied short rate models are for example the Vasicek model \cite{38},
the Hull--White model \cite{29} or the Cox--Ingersoll--Ross (CIR) model \cite{10}. In
\cite{38} and \cite{29} the short rate process is modeled by a stochastic
differential equation\index{stochastic differential equation (SDE)} (SDE) of Ornstein--Uhlenbeck (OU) type driven by a
Brownian motion (BM). As a consequence, the short rate process is normally
distributed in these models and may become arbitrarily negative. Both
features embody severe disadvantages with view on real-world market
behavior, as the distribution of interest rate\index{interest rate} data frequently deviates
from the normal distribution, while interest rates\index{interest rate} do not take arbitrarily
large negative values in practice. In the recent years, there indeed
appeared negative interest rates\index{interest rate} from time to time, but the negative values
usually were small and stayed above some lower bound. However, in \cite{10} the
short rate is modeled by a so-called square-root process. This approach
leads to a mean-reverting, strictly positive and chi-square distributed
short rate process. In \cite{6} the authors propose a time-homogeneous short
rate model which is extended by a deterministic shift function in order to
allow for negative rates\index{negative rates} and a perfect fit to the initially observed term
structure.\index{term structure} A very detailed overview on short rate models and their
properties can be found in \cite{7,16} and \cite{21}. In \cite{26} and \cite{33} short rate
models in an extended multiple-curve framework are presented. The probably
most famous forward rate\index{forward rate} model is the Heath--Jarrow--Morton (HJM) model
proposed in \cite{27}. Therein, the instantaneous forward rate\index{instantaneous forward rate} process is
modeled directly by an arithmetic SDE driven by a BM. In \cite{4} the HJM model
is extended to a jump-diffusion setup where the forward rate\index{forward rate} process is
affected by both diffusion and random jump noise. HJM type models are also
treated in \cite{7} and Chapter 7 in \cite{28}. In \cite{12} and \cite{26} HJM forward rate\index{HJM forward rate}
models in an extended multi-curve framework are discussed. The class of the so-called market models\index{market models} was introduced in \cite{5}. For example, the popular LIBOR model
belongs to this modeling class. In most cases, market models involve
geometric SDEs such that the modeled interest rates\index{interest rate} usually turn out to be
strictly positive. In order to allow the modeled rates also to take small
negative values, shifted market model approaches have been proposed
recently. Numerous properties of affine LIBOR models are provided in \cite{31}.
Market models\index{market models} are also presented in \cite{7}. In \cite{26} LIBOR models in an
extended multi-curve framework are discussed. In \cite{17} the authors propose a
L\'{e}vy forward price model in a multi-curve setup which is able to
generate negative interest rates.\index{interest rate} Term structure\index{term structure} models which are driven by
L\'{e}vy processes have also been proposed in \cite{18,19,20}.

In the present paper, we introduce a new pure-jump multi-factor short rate
model which is bounded from below by a real-valued function of time which
can be chosen arbitrarily. The short rate process is modeled by a
deterministic function plus a sum of pure-jump zero-reverting
Ornstein--Uhlenbeck processes. It turns out that the short rate is
mean-reverting and that the related bond price formula possesses an affine
representation. We also provide the dynamics of the related instantaneous
forward rate,\index{instantaneous forward rate} the latter being of HJM type. We further derive a condition
under which the forward rate\index{forward rate} model can be market-consistently calibrated.
The analytical tractability of our model is illustrated by the derivation
of a plain-vanilla option price formula with Fourier transform methods.
With view on practical applications, we make concrete assumptions on the
distribution of the jump noises and show how explicit formulas can be
deduced in these cases. We conclude the paper by presenting a post-crisis
extension of our short and forward rate\index{forward rate} model.

The outline of the paper is as follows: In Section~\ref{sec2} we introduce our new pure-jump multi-factor short rate model which is bounded from below. Section~\ref{sec3} is dedicated to the derivation of related bond price and forward rate\index{forward rate} representations. Section~\ref{sec4} is devoted to option pricing. Section~\ref{sec5} contains guidelines for a practical application, while putting a special focus on possible distributional choices for the modeling of the involved jump noises. In Section~\ref{sec6} we consider a post-crisis extension of the proposed short and forward rate\index{forward rate} model.

\vspace*{-3pt}
%s2 #&#
\section{A pure-jump multi-factor short rate model}\label{sec2}
\vspace*{-3pt}

Let $\left ( \Omega , \mathbb{F}, \mathcal{F} = \left ( \mathcal{F}_{t}
\right )_{t\in \left [ 0, T \right ]}, \mathbb{Q} \right )$ be a filtered
probability space satisfying the usual hypotheses, i.e. $\mathcal{F}_{t} =
\mathcal{F}_{t +} := \cap _{s > t} \mathcal{F}_{s}$ constitutes a
right-continuous filtration and $\mathbb{F}$ denotes the sigma-algebra
augmented by all $\mathbb{Q}$-null sets (cf. \cite{30,35}). Here,
$\mathbb{Q}$ is a risk-neutral probability measure and $T >0$ denotes a
fixed finite time horizon. In this setup, for arbitrary $n \in \mathbb{N}$
we define the stochastic short rate process $r = \left ( r_{t} \right )_{t\in
\left [ 0, T \right ]}$ via
\begin{equation}
\label{eq2.1}
r_{t} :=\mu \left ( t \right ) + \sum _{k =1}^{n} X_{t}^{k}
\end{equation}
where $\mu \left ( t \right )$ is a differentiable real-valued deterministic
$\mathcal{L}^{1}$-function and $X_{t}^{k}$ constitute pure-jump
zero-reverting Ornstein--Uhlenbeck (OU) processes satisfying the SDE
\begin{equation}
\label{eq2.2}
d X_{t}^{k} = - \lambda _{k} X_{t}^{k} dt + \sigma _{k} d L_{t}^{k}
\end{equation}
with deterministic initial values $X_{0}^{k} := x_{k} \geq 0$,
constant mean-reversion velocities $\lambda _{k} >0$ and constant volatility
coefficients $\sigma _{k} >0$. Herein, the independent, c\`{a}dl\`{a}g,
increasing, pure-jump, compound Poisson L\'{e}vy processes $L_{t}^{k}$ are
defined by
\begin{equation}
\label{eq2.3}
L_{t}^{k} := \int _{0}^{t} \int _{D_{k}} z d N_{k} \left ( s, z
\right )
\end{equation}
where $D_{k} \subseteq \mathbb{R}^{+} := \left ] 0, \infty \right [
\subset \mathbb{R}$ denote jump amplitude sets and $N_{k}$ constitute
Poisson random measures\index{Poisson random measure} (PRMs). Note that the processes $X_{t}^{k}$ and
$L_{t}^{k}$ always jump simultaneously, while $X_{t}^{k}$ decays
exponentially between its jumps due to the dampening linear drift term
appearing in (\ref{eq2.2}). A typical trajectory of a L\'{e}vy-driven OU process is
shown in Figure 15.1 in \cite{9}. Further note that the background-driving
time-homogeneous L\'{e}vy processes $L_{t}^{k}$ are increasing and thus,
constitute so-called subordinators. Moreover, for all $k\in \left \{  1,
\dots , n \right \}  $ and $\left ( s, z \right ) \in \left [ 0, T \right ] \times
D_{k}$ we define the $\mathbb{Q}$-compensated PRMs
\begin{equation}
\label{eq2.4}
d \tilde{N}_{k}^{\mathbb{Q}} \left ( s, z \right ) :=d N_{k} \left (
s, z \right ) -d \nu _{k} \left ( z \right ) ds
\end{equation}
which constitute $\left ( \mathcal{F}, \mathbb{Q} \right )$-martingale
integrators. Herein, the positive and $\sigma $-finite L\'{e}vy measures
$\nu _{k}$ satisfy the integrability conditions
\begin{equation}
\label{eq2.5}
\int _{D_{k}} \left ( 1 \wedge z \right ) d \nu _{k} \left ( z \right ) < \infty , \qquad
\int _{z >1} e^{\varpi z} d \nu _{k} \left ( z \right ) < \infty
\end{equation}
for an arbitrary constant $\varpi \in \mathbb{R}$ (cf. \cite{9,17}). For all
$k\in \left \{  1, \dots , n \right \}  $ and $t\in \left [ 0, T \right ]$ we
obtain
\begin{equation*}
\mathbb{E}_{\mathbb{Q}}  \bigl[ L_{t}^{k}  \bigr] = t \int _{D_{k}} z d
\nu _{k} \left ( z \right ), \qquad \mathbh{Var}_{\mathbb{Q}}  \bigl[ L_{t}^{k}
 \bigr] = t \int _{D_{k}} z^{2} d \nu _{k} \left ( z \right )\vadjust{\eject}
\end{equation*}
both being finite entities due to (\ref{eq2.5}) (cf. Section 1 in \cite{17}). We remark
that the currently proposed multi-factor short rate model (\ref{eq2.1}) has been
inspired by the electricity spot price model introduced in \cite{2}. Arithmetic
multi-factor models of this type have also been investigated in \cite{28} and
Section 3.2.2 in \cite{3}.

\begin{Remark}\label{rem2.1} \textbf{(a)} Since $L_{t}^{k}$ is increasing and $X_{t}^{k}$ is
zero-reverting from above, the function $\mu \left ( t \right )$ is the
mean-reversion floor or lower bound of the short rate process $r_{t}$, i.e.
it holds $r_{t} \geq \mu \left ( t \right )\ \mathbb{Q}$-a.s. for all $t\in
\left [ 0, T \right ]$, while $r_{t}$ is mean-reverting from above to $\mu
\left ( t \right )$. Also note that the presence of a Brownian motion (BM) as
driving noise in one of the processes $X_{t}^{1}, \dots , X_{t}^{n}$ would
destroy the lower boundedness of $r_{t}$. In contrast to the presented
pure-jump approach, it appears difficult to set up (lower-) bounded
processes in arithmetic BM approaches. Moreover, we recall that negative
rates\index{negative rates} have been observed in real-world post-crisis interest rate\index{interest rate} markets.
Such scenarios can easily be captured by our model by choosing, e.g. $\mu
\left ( t \right ) \equiv c$, where $c <0$ is an arbitrary constant. (In
practical applications, it may happen that the floor function $\mu \left ( t
\right )$ needs to be readjusted, if interest rates\index{interest rate} evolve lower than
anticipated. This issue has been discussed in \cite{1} in the context of the
SABR model.)

\textbf{(b)} Our pure-jump model (\ref{eq2.1}) is able to generate short rate
trajectories which closely resemble those stemming from common Brownian
motion approaches, if we allow for small jump sizes only, i.e. $D_{k} =
\left [ \epsilon _{1}^{k}, \epsilon _{2}^{k} \right ]$ with small constants $0<
\epsilon _{1}^{k} < \epsilon _{2}^{k}$. In this context, we emphasize that
the well-established pure-jump variance gamma model is likewise able to
generate suitable price trajectories, although there is neither any
diffusion component involved (cf. Section 2.6.3 in \cite{3}, Table~4.5 in \cite{9},
Section 5.3.7 in \cite{37}). On top of that, our pure-jump model might even
provide more flexibility concerning the modeling of distributional
properties than common BM approaches, since we are able to implement
tailor-made distributions via an appropriate choice of the L\'{e}vy
measures $\nu _{k}$ which fit the empirical behavior of the rates in a best
possible manner. This topic is further discussed in Section~\ref{sec5} below. For
instance, (generalized) inverse Gaussian, tempered stable or gamma
distributions might embody suitable choices (recall Appendix B.1.2 on p.
151 in \cite{37}). We finally recall that a model of the type (\ref{eq2.1}) has been
fitted to real market data in \cite{2} (yet in an electricity market context).
\end{Remark}

For a time partition $0 \leq t\leq s\leq T$ the solution of (\ref{eq2.2}) under
$\mathbb{Q}$
%points out as
can be expressed as
\begin{equation}
\label{eq2.6}
X_{s}^{k} = X_{t}^{k} e^{- \lambda _{k} \left ( s-t \right )} + \sigma _{k}
\int _{t}^{s} \int _{D_{k}} e^{- \lambda _{k} \left ( s-u \right )} z d N_{k}
\left ( u, z \right )
\end{equation}
where we used (\ref{eq2.3}). The representation (\ref{eq2.6}) implies
\begin{equation}
\label{eq2.7}
X_{t}^{k} = x_{k} e^{- \lambda _{k} t} + \sigma _{k} \int _{0}^{t}
\int _{D_{k}} e^{- \lambda _{k} \left ( t-s \right )} z d N_{k} \left ( s, z
\right )
\end{equation}
where $0 \leq t\leq T$. For all $t\in \left [ 0, T \right ]$ we next define
the historical filtration
\begin{equation*}
\mathcal{F}_{t} :=\sigma  \{  L_{s}^{1}, \dots , L_{s}^{n}:0
\leq s\leq t  \}  .
\end{equation*}

\begin{Proposition}\label{prop2.2} {For} $0 \leq u\leq t\leq T$ {we
have}
\begin{gather*}
\mathbb{E}_{\mathbb{Q}} \left ( r_{t} | \mathcal{F}_{u} \right ) = \mu
\left ( t \right ) + \sum _{k =1}^{n} \left ( X_{u}^{k} e^{- \lambda _{k} \left (
t-u \right )} + \sigma _{k} \frac{1 - e^{- \lambda _{k} \left ( t-u \right )}}{\lambda _{k}}
\int _{D_{k}} z d \nu _{k} \left ( z \right ) \right ),
\\
\mathbh{Var}_{\mathbb{Q}} \left ( r_{t} | \mathcal{F}_{u} \right ) =
\sum _{k =1}^{n} \sigma _{k}^{2} \frac{1 - e^{- 2 \lambda _{k} \left ( t-u
\right )}}{2 \lambda _{k}} \int _{D_{k}} z^{2} d \nu _{k} \left ( z \right )
\end{gather*}
{where the short rate process} $r_{t}$ {satisfies} (\ref{eq2.1}).
{Both entities are finite due to} (\ref{eq2.5}).
\end{Proposition}
(Here and in what follows,
%the following,
we omit all proofs which are
straightforward.) Taking $u =0$ in Proposition~\ref{prop2.2}, we find for all $t\in \left [ 0, T
\right ]$
\begin{gather*}
\mathbb{E}_{\mathbb{Q}} \left [ r_{t} \right ] = \mu \left ( t \right ) +
\sum _{k =1}^{n} \left ( x_{k} e^{- \lambda _{k} t} + \sigma _{k} \frac{1 -
e^{- \lambda _{k} t}}{\lambda _{k}} \int _{D_{k}} z d \nu _{k} \left ( z
\right ) \right ),
\\
\mathbh{Var}_{\mathbb{Q}} \left [ r_{t} \right ] = \sum _{k =1}^{n}
\sigma _{k}^{2} \frac{1 - e^{- 2 \lambda _{k} t}}{2 \lambda _{k}}
\int _{D_{k}} z^{2} d \nu _{k} \left ( z \right ).
\end{gather*}Note that it is
%easily
possible to identify the entities
$\mathbb{E}_{\mathbb{Q}} \left [ X_{t}^{k} \right ]$ and
$\mathbh{Var}_{\mathbb{Q}} \left [ X_{t}^{k} \right ]$ inside the latter
equations due to (\ref{eq2.1}). Moreover, suppose that $\mu \left ( t \right )
\rightarrow \tilde{\mu }$ for $t\rightarrow \infty $ where $\tilde{\mu }
\in \mathbb{R}$ is a finite constant. Then we observe
\begin{align*}
&\lim _{t\rightarrow \infty } \mathbb{E}_{\mathbb{Q}} \left [ r_{t} \right ] =
\tilde{\mu } + \sum _{k =1}^{n} \frac{\sigma _{k}}{\lambda _{k}} \int _{D_{k}}
z d \nu _{k} \left ( z \right ),
\\
&  \lim _{t\rightarrow \infty }
\mathbh{Var}_{\mathbb{Q}} \left [ r_{t} \right ] = \sum _{k =1}^{n}
\frac{\sigma _{k}^{2}}{2 \lambda _{k}} \int _{D_{k}} z^{2} d \nu _{k} \left ( z
\right )
\end{align*}
which both constitute finite constants. This limit behavior entirely stands
in line with the requirements imposed on short rate models claimed on p. 46
in \cite{7}.
In the next step, we investigate the characteristic function of $r_{t}$
which is defined via
\begin{equation*}
\Phi _{r_{t}} \left ( u \right ) := \mathbb{E}_{\mathbb{Q}}  \bigl[
e^{iu r_{t}}  \bigr]
\end{equation*}
where $u \in \mathbb{R}$ and $t\in \left [ 0, T \right ]$.

\begin{Proposition}\label{prop2.3} {For} $k\in \left \{  1, \dots , n \right \}  $
{we define the deterministic functions}
\begin{align*}
&\Lambda _{k} \left ( s, z \right ) := \sigma _{k} e^{- \lambda _{k}
\left ( t-s \right )} z,\qquad  \psi _{k} \left ( t, u \right ) :=iu e^{-
\lambda _{k} t},
\\
& \rho _{k} \left ( t, u \right ) := \int _{0}^{t}
\int _{D_{k}}  \bigl[ e^{iu \Lambda _{k} \left ( s, z \right )} - 1  \bigr] d
\nu _{k} \left ( z \right ) ds.
\end{align*}
{Then for any} $u \in \mathbb{R}$ {and} $t\in \left [ 0, T
\right ]$ {the characteristic function of} $r_{t}$ {can be
decomposed as}
\begin{equation*}
\Phi _{r_{t}} \left ( u \right ) = e^{iu\mu \left ( t \right )} \prod _{k=1}^{n}
\Phi _{X_{t}^{k}} \left ( u \right )
\end{equation*}
{where the characteristic function of} $X_{t}^{k}$ {is given
by}
\begin{equation*}
\Phi _{X_{t}^{k}} \left ( u \right ) = e^{\psi _{k} \left ( t,u \right ) x_{k} +
\rho _{k} \left ( t,u \right )}
\end{equation*}
{with deterministic and affine characteristic exponent}.
\end{Proposition}

An immediate consequence of Proposition~\ref{prop2.3} is the subsequent affine
representation
\begin{equation}
\label{eq2.8}
\Phi _{r_{t}} \left ( u \right ) = \prod _{k=1}^{n} e^{\psi _{k} \left ( t,u
\right ) x_{k} + \phi _{k} \left ( t,u \right )}
\end{equation}
where we introduced the deterministic functions
\begin{equation*}
\phi _{k} \left ( t,u \right ) := \rho _{k} \left ( t,u \right ) +iu
{\mu \left ( t \right )} / {n}.
\end{equation*}We emphasize that $r_{t}$ is an affine function of the factors $X_{t}^{1},
\dots , X_{t}^{n}$ such that our model turns out to be a special case of the
affine short rate models considered in Section 3.3 in \cite{14}. To read more on
affine processes we refer to \cite{7,14,16,26} and \cite{31}. We next define
the moment generating function of $r_{t}$ via
\begin{equation}
\label{eq2.9}
\kappa _{r_{t}} \left ( v \right ) := \mathbb{E}_{\mathbb{Q}} \left [
e^{v r_{t}} \right ]
\end{equation}
which implies the well-known equalities $\Phi _{r_{t}} \left ( u \right ) =
\kappa _{r_{t}} \left ( iu \right )$ and $\kappa _{r_{t}} \left ( v \right ) =
\Phi _{r_{t}} \left ( - iv \right )$. Note that the moment generating function
$\kappa _{r_{t}} \left ( v \right )$ is well-defined due to (\ref{eq2.5}). In the
sequel, we derive the time dynamics of the short rate process.

\begin{Proposition}\label{prop2.4} {For all} $t \in \left [ 0,T \right ]$
{the short rate process follows the dynamics}
\begin{equation}
\label{eq2.10}
d r_{t} = \left ( \mu ' \left ( t \right ) - \sum _{k=1}^{n} \lambda _{k}
X_{t}^{k} \right ) dt+ \sum _{k=1}^{n} \sigma _{k} \int _{D_{k}} z d N_{k}
\left ( t,z \right ).
\end{equation}
\end{Proposition}

\begin{Remark}\label{rem2.5} We recall that our model constitutes an extension of
the short rate model proposed in \cite{6}, whereas we work with multiple
pure-jump processes $L_{t}^{1}, \dots , L_{t}^{n}$ as driving noises instead
of the single Brownian motion $W_{t}$ appearing in Eq. (1) in \cite{6}.
Moreover, comparing Eq. (3) in \cite{6} with Eq.~(\ref{eq2.1}) above, we see that
$x_{t}$ and $\varphi \left ( t;\alpha \right )$ in \cite{6} correspond in our
setup to $\sum _{k=1}^{n} X_{t}^{k}$ and $\mu \left ( t \right )$,
respectively.
\end{Remark}

%s3 #&#
\section{Bond prices and forward rates\index{forward rate}}\label{sec3}

In this section, we derive representations for zero-coupon bond prices,
forward rates\index{forward rate} and the interest rate\index{interest rate} curve related to the short rate model
introduced in Section~\ref{sec2}. To begin with, we introduce a bank account with
stochastic interest rate\index{interest rate} $r_{t}$ satisfying
\begin{equation}
\label{eq3.1}
d \beta _{t} = r_{t} \beta _{t} dt
\end{equation}
with normalized initial capital $\beta _{0} =1$. The solution of (\ref{eq3.1}) reads
as
\begin{equation}
\label{eq3.2}
\beta _{t} = \exp \left \{  \int _{0}^{t} r_{s} ds \right \}
\end{equation}
where $t \in \left [ 0,T \right ]$. In this setup, the (zero-coupon) bond
price at time $t \leq T$ with maturity $T$ is given by
\begin{equation}
\label{eq3.3}
P \left ( t,T \right ) := \beta _{t} \mathbb{E}_{\mathbb{Q}} \left (
\beta _{T}^{-1} | \mathcal{F}_{t} \right ) = \mathbb{E}_{\mathbb{Q}}
\left ( \exp \left \{  - \int _{t}^{T} r_{s} ds \right \}  \bigg| \mathcal{F}_{t}
\right )
\end{equation}
where $t \in \left [ 0,T \right ]$ (cf. \cite{6,7,26}). Note that $P \left (
t,T \right ) >0\ \mathbb{Q}$-a.s. $\forall \ t \in \left [ 0,T \right ]$ by
construction. Since $r_{t} \geq \mu \left ( t \right )\ \mathbb{Q}$-a.s.
$\forall \ t \in \left [ 0,T \right ]$ [recall Remark~\ref{rem2.1} (a)], we
observe
\begin{equation}
\label{eq3.4}
P \left ( t,T \right ) \leq M_{t,T} := \exp \left \{  - \int _{t}^{T}
\mu \left ( s \right ) ds \right \}
\end{equation}
$\mathbb{Q}$-a.s. $\forall \ t \in \left [ 0,T \right ]$ due to (\ref{eq3.3}) and
the monotonicity of conditional expectations. The upper bound $M_{t,T}$
appearing in (\ref{eq3.4}) is deterministic and strictly positive for all $0 \leq t
\leq T$. If $\mu \left ( t \right ) \geq 0$, then it holds $P \left ( t,T
\right ) \leq 1\ \mathbb{Q}$-a.s. $\forall \ t \in \left [ 0,T \right ]$
(similar to, e.g., the CIR model \cite{10}; also see \cite{7,21}). On the other
hand, if $\mu \left ( t \right ) <0$, then we only know that $M_{t,T} >1$.

\begin{Proposition}\label{prop3.1} {For} $k\in \left \{  1,\dots ,n \right \}  $
{and} $t\in \left [ 0,T \right ]$ {we define the deterministic
functions}
\begin{equation}
\label{eq3.5}
\begin{split}
&A_{k} \left ( t,T \right ) := \int _{t}^{T} \left ( - \frac{\mu \left (
s \right )}{n} + \int _{D_{k}}  \bigl[ e^{\sigma _{k} B_{k} \left ( s,T \right )
z} -1  \bigr] d \nu _{k} \left ( z \right ) \right ) ds,
\\
& B_{k} \left ( t,T
\right ) := \frac{e^{- \lambda _{k} \left ( T-t \right )} -1}{\lambda _{k}} \leq 0.
\end{split}
\end{equation}Then the bond price at time $t\leq T$ {with maturity} $T$
{possesses the affine representation}
\begin{equation}
\label{eq3.6}
P \left ( t,T \right ) = \prod _{k=1}^{n} e^{A_{k} \left ( t,T \right ) + B_{k}
\left ( t,T \right ) X_{t}^{k}}
\end{equation}
{where the factors} $X_{t}^{k}$ {satisfy} (\ref{eq2.7}).
\end{Proposition}

\begin{proof} First of all, we put (\ref{eq2.6}) into (\ref{eq2.1}) and obtain
\begin{equation*}
r_{s} =\mu \left ( s \right ) + \sum _{k=1}^{n} X_{t}^{k} e^{- \lambda _{k}
\left ( s-t \right )} + \sum _{k=1}^{n} \int _{t}^{s} \int _{D_{k}} \sigma _{k}
e^{- \lambda _{k} \left ( s-u \right )} z d N_{k} \left ( u,z \right )
\end{equation*}
where $0\leq t\leq s\leq T$. We next substitute the latter equation into
(\ref{eq3.3}), hereafter apply Fubini's theorem and identify the functions $B_{k}
\left ( \boldsymbol{\cdot },T \right )$. This procedure yields
\begin{align*}
P \left ( t,T \right ) &= \exp \left \{  - \int _{t}^{T} \mu \left ( s \right ) ds
+ \sum _{k=1}^{n} B_{k} \left ( t,T \right ) X_{t}^{k} \right \}
\\
&\quad   \times
\mathbb{E}_{\mathbb{Q}} \left ( \exp \left \{  \sum _{k=1}^{n} \int _{t}^{T}
\int _{D_{k}} \sigma _{k} B_{k} \left ( s,T \right ) z d N_{k} \left ( s,z
\right ) \right \}  \bigg| \mathcal{F}_{t} \right ).
\end{align*}Taking the independent increment property of the ($\mathbb{Q}$-independent)
L\'{e}vy processes $L^{1},\dots , L^{n}$ into account, we obtain
\begin{align*}
&\mathbb{E}_{\mathbb{Q}} \left ( \exp \left \{  \sum _{k=1}^{n} \int _{t}^{T}
\int _{D_{k}} \sigma _{k} B_{k} \left ( s,T \right ) z d N_{k} \left ( s,z
\right ) \right \}  \bigg| \mathcal{F}_{t} \right )
\\
&\quad  = \prod _{k=1}^{n}
\mathbb{E}_{\mathbb{Q}} \left [ \exp \left \{  \int _{t}^{T} \int _{D_{k}}
\sigma _{k} B_{k} \left ( s,T \right ) z d N_{k} \left ( s,z \right ) \right \}
\right ]
\end{align*}
where $t\in \left [ 0,T \right ]$. The %appearing
usual expectations appearing here can be
handled by the L\'{e}vy--Khinchin formula for additive processes (see, e.g.,
\cite{9,30,36}) which leads us to
\begin{align*}
&\mathbb{E}_{\mathbb{Q}} \left [ \exp \left \{  \int _{t}^{T} \int _{D_{k}}
\sigma _{k} B_{k} \left ( s,T \right ) z d N_{k} \left ( s,z \right ) \right \}
\right ]
\\
&\quad  = \exp \left \{  \int _{t}^{T} \int _{D_{k}}  \bigl[ e^{\sigma _{k} B_{k}
\left ( s,T \right ) z} -1  \bigr] d \nu _{k} \left ( z \right ) ds \right \}  .
\end{align*}Putting the latter equations together and identifying the functions $A_{k}
\left ( \boldsymbol{\cdot },T \right )$, we end up with the asserted
representation (\ref{eq3.6}).
\end{proof}

Recall that the bond price in (\ref{eq3.6}) is the product of exponential affine
functions of the factors $X_{t}^{1},\dots , X_{t}^{n}$ (but not of $r_{t}$).
Also note that for all $k\in \left \{  1,\dots ,n \right \}  $ and $t\in \left [
0,T \right ]$ it holds
\begin{equation}
\label{eq3.7}
A_{k} \left ( t,t \right ) = B_{k} \left ( t,t \right ) =0.
\end{equation}We remark that the functions $B_{k} \left ( t,T \right )$ in (\ref{eq3.5}) possess
the same structure as the corresponding ones in the Vasicek model (cf.
\cite{38}, or \cite{7,16,21}). For all $t\in \left [ 0,T \right ]$ Eq.~(\ref{eq3.6}) can
be rewritten as
\begin{equation}
\label{eq3.8}
P \left ( t,T \right ) = \exp \left \{  \sum _{k=1}^{n} \left [ A_{k} \left ( t,T
\right ) + B_{k} \left ( t,T \right ) X_{t}^{k} \right ] \right \}
\end{equation}
which implies $P \left ( T,T \right ) =1$ due to (\ref{eq3.7}). Moreover, from (\ref{eq3.5})
we infer the time derivatives
\begin{equation}
\label{eq3.9}
A_{k} ' \left ( t,T \right ) = \frac{\mu \left ( t \right )}{n} - \int _{D_{k}}
 \bigl[ e^{\sigma _{k} B_{k} \left ( t,T \right ) z} -1  \bigr] d \nu _{k}
\left ( z \right ),\qquad  B_{k} ' \left ( t,T \right ) = e^{- \lambda _{k} \left ( T-t
\right )} >0
\end{equation}
where $A_{k} ' := \partial _{t} A_{k}$ and $B_{k} ' :=
\partial _{t} B_{k}$. Hence, the functions $B_{k} \left ( t,T \right ) \leq 0$
are strictly monotone increasing in $t$. Also note that the formulas found
in (\ref{eq3.9}) entirely stand in line with those claimed in (4.4)--(4.5) in \cite{31}.
From (\ref{eq3.5}), (\ref{eq3.7}) and (\ref{eq3.9}) we deduce the following system of ordinary
differential equations\index{ordinary differential equations (ODE)} (ODEs)
\begin{align*}
&A_{k} \left ( t,T \right ) =- \int _{t}^{T} A_{k} ' \left ( s,T \right ) ds,\qquad
B_{k} ' \left ( t,T \right ) =1+ \lambda _{k} B_{k} \left ( t,T \right ),
\\
& A_{k}
\left ( T,T \right ) = B_{k} \left ( T,T \right ) =0
\end{align*}
where $t\in \left [ 0,T \right ]$ and $k\in \left \{  1,\dots ,n \right \}  $. We
are now prepared to derive the time dynamics of the bond price process
$\left ( P \left ( t,T \right ) \right )_{t\in \left [ 0,T \right ]}$.

\begin{Proposition}\label{prop3.2} {For} $k\in \left \{  1,\dots ,n \right \}  $,
$t\in \left [ 0,T \right ]$ {and} $z\in D_{k}$ {we define the
deterministic functions}
\begin{equation}
\label{eq3.10}
\zeta _{k} \left ( t,T,z \right ) := e^{\sigma _{k} B_{k} \left ( t,T
\right ) z} -1
\end{equation}
{with} $B_{k} \left ( t,T \right )$ {as in} (\ref{eq3.5}).
{Then the bond price satisfies the} $t$-dynamics under
$\mathbb{Q}$
\begin{equation}
\label{eq3.11}
\frac{dP \left ( t,T \right )}{P \left ( t-,T \right )} = r_{t} dt+
\sum _{k=1}^{n} \int _{D_{k}} \zeta _{k} \left ( t,T,z \right ) d
\tilde{N}_{k}^{\mathbb{Q}} \left ( t,z \right ).
\end{equation}
\end{Proposition}

Recall that it holds $\zeta _{k} \left ( t,T,z \right ) \leq 0$, since
$\sigma _{k} \ B_{k} \left ( t,T \right ) \ z\leq 0$ for all $k$, $t$ and $z$.
We stress that (\ref{eq3.11}) possesses the same structure as the corresponding Eq.
(10.9) in \cite{37}, whereas the latter stems from a Brownian motion model
without jumps. In the next step, we provide the solution of the SDE
(\ref{eq3.11}).

\begin{Proposition}\label{prop3.3} {For all} $t\in \left [ 0,T \right ]$
{the solution of} (\ref{eq3.11}) {under} $\mathbb{Q}$
{reads as}
\begin{align}
\label{eq3.12}
P \left ( t,T \right ) &=P \left ( 0,T \right ) \exp  \Biggl\{  \int _{0}^{t} r_{s}
ds - \sum _{k=1}^{n} \int _{0}^{t} \int _{D_{k}} \zeta _{k} \left ( s,T,z
\right ) d \nu _{k} \left ( z \right ) ds
\nonumber\\
&\quad + \sum _{k=1}^{n} \int _{0}^{t}
\int _{D_{k}} \sigma _{k} B_{k} \left ( s,T \right ) z d N_{k} \left ( s,z
\right ) \Biggr\}
\end{align}
{where the initial value} $P \left ( 0,T \right )$ {is
deterministic and fulfills} $P \left ( 0,T \right ) >0$.
\end{Proposition}

Furthermore, for all $t\in \left [ 0,T \right ]$ let us introduce the
discontinuous Dol\'{e}ans-Dade exponential
\begin{align}
\label{eq3.13}
\Xi _{t}^{k}& := \mathcal{E}  \bigl( h_{k} *
\tilde{N}_{k}^{\mathbb{Q}}  \bigr)_{t} := \exp \Biggl\{
\int _{0}^{t} \int _{D_{k}} h_{k} \left ( s,z \right ) d
\tilde{N}_{k}^{\mathbb{Q}} \left ( s,z \right )
\nonumber\\
&\quad  - \int _{0}^{t} \int _{D_{k}}
 \bigl[ e^{h_{k} \left ( s,z \right )} -1- h_{k} \left ( s,z \right )  \bigr] d
\nu _{k} \left ( z \right ) ds \Biggr\}
\end{align}
where $h_{k} \left ( s,z \right )$ is an arbitrary integrable deterministic
function (which may also depend on $T$). We recall that $\Xi _{0}^{k} =1$
and that $\left ( \Xi _{t}^{k} \right )_{t\in \left [ 0,T \right ]}$ constitutes
a local $\mathbb{Q}$-martingale. With definition~(\ref{eq3.13}) at hand, we can
express Eq.~(\ref{eq3.12}) as follows.

\begin{Corollary}\label{cor3.4} {For all} $0\leq t\leq T$ {the bond
price satisfies}
\begin{equation}
\label{eq3.14}
P \left ( t,T \right ) =P \left ( 0,T \right ) \beta _{t} \prod _{k=1}^{n}
\mathcal{E}  \bigl( \xi _{k} * \tilde{N}_{k}^{\mathbb{Q}}  \bigr)_{t}
\end{equation}
{where} $\beta _{t}$ {is the bank account process given in}
(\ref{eq3.2}), $\mathcal{E}$ {denotes the Dole\'{a}ns-Dade exponential
defined in} (\ref{eq3.13}) {and} $\xi _{k} \left ( s,z \right ) :=
\sigma _{k} \ B_{k} \left ( s,T \right ) \ z= \log \left ( 1+ \zeta _{k} \left (
s,T,z \right ) \right )$ {is a deterministic function}.
\end{Corollary}

Moreover, for all $0\leq t\leq T$ we define the discounted bond
price\index{discounted bond price}
\begin{equation}
\label{eq3.15}
\hat{P} \left ( t,T \right ) := \frac{P \left ( t,T \right )}{\beta _{t}}
\end{equation}
where $\hat{P} \left ( 0,T \right ) =P \left ( 0,T \right )$. From (\ref{eq3.3}) we
deduce $\hat{P} \left ( t,T \right ) = \mathbb{E}_{\mathbb{Q}}  (
\beta _{T}^{-1} | \ \mathcal{F}_{t}  )$ such that $\hat{P} \left ( t,T
\right )$ constitutes an $\mathcal{F}_{t}$-adapted (true) martingale under
$\mathbb{Q}$, as required by the risk-neutral pricing theory. Plugging
(\ref{eq3.14}) into (\ref{eq3.15}), for all $t\in \left [ 0,T \right ]$ we obtain
\begin{equation}
\label{eq3.16}
\hat{P} \left ( t,T \right ) =P \left ( 0,T \right ) \prod _{k=1}^{n}
\mathcal{E}  \bigl( \xi _{k} * \tilde{N}_{k}^{\mathbb{Q}}  \bigr)_{t}\vadjust{\eject}
\end{equation}
where $P \left ( 0,T \right )$ is deterministic and $\xi _{k}$ is such as
defined in Corollary~\ref{cor3.4}. We obtain the following result.

\begingroup
\abovedisplayskip=4pt\belowdisplayskip=4pt
\begin{Proposition}\label{prop3.5} {For all} $t\in \left [ 0,T \right ]$
{the discounted bond price\index{discounted bond price} satisfies the}
$\mathbb{Q}$-mar\-tingale dynamics
\begin{equation*}
\frac{d \hat{P} \left ( t,T \right )}{\hat{P} \left ( t-,T \right )} =
\sum _{k=1}^{n} \int _{D_{k}} \zeta _{k} \left ( t,T,z \right ) d
\tilde{N}_{k}^{\mathbb{Q}} \left ( t,z \right )
\end{equation*}
{where the deterministic functions} $\zeta _{k} \left ( t,T,z \right )$
{are such as defined in} (\ref{eq3.10}).
\end{Proposition}

With reference to \cite{7}, we define the instantaneous forward rate\index{instantaneous forward rate} at time $t$
with maturity $T$ via
\begin{equation}
\label{eq3.17}
f \left ( t,T \right ) :=- \partial _{T} \log P \left ( t,T \right )
\end{equation}
where $t\in \left [ 0,T \right ]$ and $\partial _{T}$ denotes the differential
operator with respect to $T$. Equation (\ref{eq3.17}) is equivalent to the
representation
\begin{equation}
\label{eq3.18}
P \left ( t,T \right ) = \exp \left \{  - \int _{t}^{T} f \left ( t,u \right ) du
\right \}  .
\end{equation}

\begin{Lemma}\label{lem3.6} {For all} $k\in \left \{  1,\dots ,n \right \}  $
{and} $t\in \left [ 0,T \right ]$ {it holds}
\begin{align*}
\partial _{T} A_{k} \left ( t,T \right ) &=- \frac{\mu \left ( T \right )}{n} -
\sigma _{k} \int _{t}^{T} \int _{D_{k}} z e^{\sigma _{k} B_{k} \left ( s,T
\right ) z- \lambda _{k} \left ( T-s \right )} d \nu _{k} \left ( z \right ) ds, \\\
\partial _{T} B_{k} \left ( t,T \right )
&=- e^{- \lambda _{k} \left ( T-t
\right )}.
\end{align*}
\end{Lemma}

\begin{proof} By the definition of $B_{k} \left ( t,T \right )$ claimed in
(\ref{eq3.5}) we find
\begin{equation}
\label{eq3.19}
\partial _{T} B_{k} \left ( t,T \right ) =- e^{- \lambda _{k} \left ( T-t
\right )}
\end{equation}
so that the functions $B_{k} \left ( t,T \right )$ are strictly monotone
decreasing in $T$. From (\ref{eq3.5}) and (\ref{eq3.10}) we further deduce
\begin{equation*}
\partial _{T} A_{k} \left ( t,T \right ) =- \frac{\mu \left ( T \right )}{n} +
\partial _{T} \left ( \int _{t}^{T} \int _{D_{k}} \zeta _{k} \left ( s,T,z
\right ) d \nu _{k} \left ( z \right ) ds \right )
\end{equation*}
whereas Fubini's theorem (see, e.g., Theorem 2.2 in \cite{3}) leads us to
\begin{equation*}
\partial _{T} \left ( \int _{t}^{T} \int _{D_{k}} \zeta _{k} \left ( s,T,z
\right ) d \nu _{k} \left ( z \right ) ds \right ) = \int _{D_{k}} \partial _{T}
\left ( \int _{t}^{T} \zeta _{k} \left ( s,T,z \right ) ds \right ) d \nu _{k}
\left ( z \right ).
\end{equation*} (We are able to apply Fubini's theorem here, since the deterministic
function $\zeta _{k}  ( s,\!T,\!z  )$ is measurable and
square-integrable with respect to $s\in \left [ 0,T \right ]$ and $z\in
D_{k}$.) The Leibniz formula for parameter integrals (see Lemma 2.4.1 on p.
13 in \cite{28}) yields
\begin{align*}
\partial _{T} \left ( \int _{t}^{T} \zeta _{k} \left ( s,T,z \right ) ds \right )
&= \zeta _{k} \left ( T,T,z \right ) + \int _{t}^{T} \partial _{T} \zeta _{k}
\left ( s,T,z \right ) ds
\\
&=- \sigma _{k} \int _{t}^{T} z e^{\sigma _{k} B_{k}
\left ( s,T \right ) z- \lambda _{k} \left ( T-s \right )} ds
\end{align*}
where we used (\ref{eq3.10}), (\ref{eq3.7}) and (\ref{eq3.19}). Putting these formulas together,
the proof is complete.
\end{proof}\eject\endgroup

\begin{Proposition}\label{prop3.7} {For all} $t\in \left [ 0,T \right ]$
{the instantaneous forward rate\index{instantaneous forward rate} can be represented as}
\begin{equation}
\label{eq3.20}
f \left ( t,T \right ) =\mu \left ( T \right ) + \sum _{k=1}^{n} \int _{t}^{T}
\int _{D_{k}} \sigma _{k} z e^{\sigma _{k} B_{k} \left ( s,T \right ) z-
\lambda _{k} \left ( T-s \right )} d \nu _{k} \left ( z \right ) ds +
\sum _{k=1}^{n} X_{t}^{k} e^{- \lambda _{k} \left ( T-t \right )}
\end{equation}
{where the factor processes} $X_{t}^{k}$ {satisfy} (\ref{eq2.7})
{and} $B_{k} \left ( s,T \right )$ {is like defined in}
(\ref{eq3.5}).
\end{Proposition}

\begin{proof} We substitute (\ref{eq3.8}) into (\ref{eq3.17}) and obtain
\begin{equation*}
f \left ( t,T \right ) =- \sum _{k=1}^{n}  \bigl[ \partial _{T} A_{k} \left ( t,T
\right ) + X_{t}^{k} \partial _{T} B_{k} \left ( t,T \right )  \bigr].
\end{equation*}Combining this equality with Lemma~\ref{lem3.6}, we derive the claimed
representation (\ref{eq3.20}).
\end{proof}

Replacing $T$ by $t$ in (\ref{eq3.20}), we immediately find $f \left ( t,t \right ) =
r_{t}$ due to (\ref{eq2.1}). This equality stands in line with the usual
conventions in interest rate\index{interest rate} theory (see, e.g., \cite{7,16,21}).

\begin{Proposition}\label{prop3.8} {For all} $t\in \left [ 0,T \right ]$
{the instantaneous forward rate\index{instantaneous forward rate} fulfills the pure-jump multi-factor
HJM type equation}
\begin{equation}
\label{eq3.21}
f \left ( t,T \right ) =f \left ( 0,T \right ) + \sum _{k=1}^{n} \int _{0}^{t}
\int _{D_{k}} \sigma _{k} z e^{- \lambda _{k} \left ( T-s \right )}  \bigl\{  d
N_{k} \left ( s,z \right ) - e^{\sigma _{k} B_{k} \left ( s,T \right ) z} d
\nu _{k} \left ( z \right ) ds \bigr\}
\end{equation}
{where the deterministic initial value is given by} $f \left ( 0,T
\right ) =- \partial _{T} \log P \left ( 0,T \right )$.
\end{Proposition}

In what follows, we illustrate how our forward rate\index{forward rate} model can be fitted to
the initially observed term structure.\index{term structure} This procedure is often called
\textit{market-consistent calibration} in the literature. For this purpose,
we denote by $f^{M} \left ( 0,T \right )$ the deterministic initial forward
rate. If $f \left ( 0,T \right ) = f^{M} \left ( 0,T \right )$ and hence, if $P
\left ( 0,T \right ) = P^{M} \left ( 0,T \right )$ for all maturity times
$T>0$, then the underlying model is called \textit{market-consistent}.

\begin{Proposition}\label{prop3.9} {The forward rate\index{forward rate} model}
(\ref{eq3.20})--(\ref{eq3.21}) {can be market-consistently calibrated to a
given term structure} $f^{M} \left ( 0,T \right )$ {by choosing the
floor function} $\mu \left ( \boldsymbol{\cdot } \right )\ in$ (\ref{eq3.20})
{according to}
\begin{equation}
\label{eq3.22}
\mu \left ( T \right ) = f^{M} \left ( 0,T \right ) - \sum _{k=1}^{n} \left (
x_{k} e^{- \lambda _{k} T} + \int _{0}^{T} \int _{D_{k}} \sigma _{k} z
e^{\sigma _{k} B_{k} \left ( s,T \right ) z- \lambda _{k} \left ( T-s \right )} d
\nu _{k} \left ( z \right ) ds \right )
\end{equation}
{for all maturity times} $T>0$.
\end{Proposition}

Note that the floor function $\mu \left ( t \right )$ for all $t\in \left [
0,T \right ]$ can be obtained from (\ref{eq3.22}) by replacing $T$ by $t$ therein.
Moreover, we define the interest rate\index{interest rate} curve at time $t<T$ with maturity $T$
via
\begin{equation}
\label{eq3.23}
R \left ( t,T \right ) := \frac{\log P \left ( t,T \right )}{t-T}.
\end{equation}This object is called continuously-compounded spot rate on p. 60 in \cite{7}. It
obviously holds
\begin{equation}
\label{eq3.24}
P \left ( t,T \right ) = e^{- \left ( T-t \right ) R \left ( t,T \right )}
\end{equation}
where $t\in \left [ 0,T \right [$. Comparing the exponent in (\ref{eq3.24}) with
that in (\ref{eq3.8}), we %can read off
infer
\begin{equation*}
R \left ( t,T \right ) = \frac{1}{t-T} \sum _{k=1}^{n}  \bigl[ A_{k} \left (
t,T \right ) + B_{k} \left ( t,T \right ) X_{t}^{k}  \bigr]
\end{equation*}
where $A_{k}$ and $B_{k}$ are such as defined in (\ref{eq3.5}). Hence, it turns
out that the interest rate\index{interest rate} curve $R \left ( t,T \right )$ can be represented
as a sum of affine functions of the pure-jump OU factors $X_{t}^{1},\dots ,
X_{t}^{n}$. In this sense, our short rate model possesses an \textit{affine
term structure\index{term structure}} (cf. Section 3.2.4 in \cite{7}, or \cite{14,16,21}). The latter
observation entirely stands in line with (\ref{eq3.8}).

\begin{Proposition}\label{prop3.10} {For all} $t\in \left [ 0,T \right [$
{the interest rate\index{interest rate} curve possesses the representation}
\begin{align*}
R \left ( t,T \right ) &= \frac{1}{t-T} \Biggl( \log P \left ( 0,T \right ) +
\int _{0}^{t} r_{s} ds - \sum _{k=1}^{n} \int _{0}^{t} \int _{D_{k}} \zeta _{k}
\left ( s,T,z \right ) d \nu _{k} \left ( z \right ) ds
\\
&\quad  + \sum _{k=1}^{n}
\int _{0}^{t} \int _{D_{k}} \sigma _{k} B_{k} \left ( s,T \right ) z d N_{k}
\left ( s,z \right )  \Biggr)
\end{align*}
{where} $\zeta _{k}$ {and} $B_{k}$ {are such as defined
in} (\ref{eq3.10}) {and} (\ref{eq3.5}), {respectively}.
\end{Proposition}

\begin{Proposition}\label{prop3.11} {For all} $t\in \left [ 0,T \right ]$
{the integrated short rate process can be represented
as}
\begin{equation}
\label{eq3.25}
I_{t} := \int _{0}^{t} r_{s} ds = \int _{0}^{t} \mu \left ( s \right )
ds - \sum _{k=1}^{n} x_{k} B_{k} \left ( 0,t \right ) - \sum _{k=1}^{n}
\int _{0}^{t} \int _{D_{k}} \sigma _{k} B_{k} \left ( s,t \right ) z d N_{k}
\left ( s,z \right )
\end{equation}
{where the deterministic functions} $B_{k}$ {are such as
defined in} (\ref{eq3.5}).
\end{Proposition}

\begin{proof} We substitute (\ref{eq2.1}) and (\ref{eq2.7}) into the definition of
$I_{t}$ and obtain
\begin{equation*}
I_{t} = \int _{0}^{t} \mu \left ( s \right ) ds - \sum _{k=1}^{n} x_{k} B_{k}
\left ( 0,t \right ) + \sum _{k=1}^{n} \int _{0}^{t} \int _{0}^{s} \int _{D_{k}}
\sigma _{k} e^{- \lambda _{k} \left ( s-u \right )} z d N_{k} \left ( u,z
\right ) ds
\end{equation*}
where $B_{k}$ is like defined in (\ref{eq3.5}). An application of Fubini's theorem
(see Theorem 2.2 in \cite{3}) yields
\begin{equation*}
\int _{0}^{t} \int _{0}^{s} \int _{D_{k}} \sigma _{k} e^{- \lambda _{k} \left (
s-u \right )} z d N_{k} \left ( u,z \right ) ds =- \int _{0}^{t} \int _{D_{k}}
\sigma _{k} B_{k} \left ( u,t \right ) z d N_{k} \left ( u,z \right ),
\end{equation*}
so that the proof is complete.
\end{proof}

Recall that the last jump integral in (\ref{eq3.25}) constitutes a so-called
Volterra integral, as the time parameter $t$ appears both inside the
integrand and inside the upper integration bound. Also note that it holds
$I_{t} = \log \beta _{t}$ with $I_{0} =0$ due to (\ref{eq3.2}).

%s4 #&#
\section{Option pricing}\label{sec4}

In this section, we investigate the evaluation of a plain vanilla option
written on the zero-coupon bond price $P \left ( \boldsymbol{\cdot },T
\right )$. With reference to the risk-neutral pricing theory, the price at
time $t\leq \tau $ of an option with payoff $H_{\tau }$ at the maturity time\index{maturity time}
$\tau $ reads as
\begin{equation}
\label{eq4.1}
C_{t} = \beta _{t} \mathbb{E}_{\mathbb{Q}}  \bigl( \beta _{\tau }^{-1} H_{\tau }
| \mathcal{F}_{t}  \bigr) = \mathbb{E}_{\mathbb{Q}}  \bigl( e^{-
\int _{t}^{\tau } r_{s} ds} H_{\tau } | \mathcal{F}_{t}  \bigr)
\end{equation}
where $\beta $ is the bank account process given in (\ref{eq3.2}) and $\mathbb{Q}$
denotes a risk-neutral pricing measure (cf. Eq. (3.1) in \cite{7}). We now
consider a call option written on the bond price $P \left (
\boldsymbol{\cdot },T \right )$ with maturity time\index{maturity time} $T$ satisfying
$T\geq \tau $. The payoff of the call option written on $P \left ( \tau ,T
\right )$ with deterministic strike price $K>0$ and maturity time $\tau $ is
then given by
\begin{equation}
\label{eq4.2}
H_{\tau } = \left [ P \left ( \tau ,T \right ) -K \right ]^{+} := \max
\left \{  0,P \left ( \tau ,T \right ) -K \right \}  .
\end{equation}In what follows, we define the Fourier transform, respectively inverse
Fourier transform,\index{inverse Fourier transform} of a real-valued deterministic function $p \left (
\boldsymbol{\cdot } \right ) \in \mathcal{L}^{1} \left ( \mathbb{R} \right )$
via
\begin{equation*}
\hat{p} \left ( y \right ) := \frac{1}{2\pi } \int _{\mathbb{R}} p
\left ( u \right ) e^{-iyu} du, p \left ( u \right ) = \int _{\mathbb{R}}
\hat{p} \left ( y \right ) e^{iyu} dy.
\end{equation*}

\begin{Proposition}\label{prop4.1}[call option on bond price] {For all}
$0\leq t\leq \tau \leq T$ {the price of a call option with payoff}
$H_{\tau }$ {given in} (\ref{eq4.2}), {strike price} $K>0$
{and maturity time\index{maturity time}} $\tau $ {can be expressed as}
\begin{align}
\label{eq4.3}
C_{t} &= \int _{\mathbb{R}} \hat{q} \left ( y \right ) \exp \Biggl\{  I_{t}
+\theta \left ( \tau ,y \right ) + \sum _{k=1}^{n} \overline{\psi }_{k} \left (
t,\tau ,y \right )
\nonumber\\
&\quad  + \sum _{k=1}^{n} \int _{0}^{t} \int _{D_{k}} \eta _{k} \left (
s,z,y \right ) d N_{k} \left ( s,z \right ) \Biggr\}  dy
\end{align}
{where the integrated short rate process} $I_{t}$ {is such as
defined in} (\ref{eq3.25}) {and}
\begin{align}
\label{eq4.4}
\eta _{k} \left ( s,z,y \right ) &:= \eta _{k} \left ( s,z,T,\tau ,y
\right )
\nonumber\\
&:= \sigma _{k} \left [ \left ( a+iy \right ) B_{k} \left ( s,T
\right ) - \left ( a+iy-1 \right ) B_{k} \left ( s,\tau \right ) \right ] z,
\nonumber \\
\hat{q} \left ( y \right ) &:= \frac{P \left ( 0,T \right )^{a+iy}}{2\pi \left ( a+iy \right ) \left ( a+iy-1 \right ) K^{a+iy-1}},
\nonumber\\
\overline{\psi }_{k} \left ( t,\tau ,y \right ) &:= \int _{t}^{\tau }
\int _{D_{k}}  \bigl[ e^{\eta _{k} \left ( s,z,y \right )} -1  \bigr] d \nu _{k}
\left ( z \right ) ds,
 \\
\theta \left ( \tau ,y \right ) &:= \left ( a+iy-1 \right ) \left (
\int _{0}^{\tau } \mu \left ( s \right ) ds - \sum _{k=1}^{n} x_{k} B_{k} \left (
0,\tau \right ) \right )
\nonumber\\
&\quad  - \left ( a+iy \right ) \sum _{k=1}^{n} \int _{0}^{\tau }
\int _{D_{k}} \zeta _{k} \left ( s,T,z \right ) d \nu _{k} \left ( z \right ) ds
\nonumber
\end{align}
{constitute deterministic functions, while} $a>1$ {is an
arbitrary real-valued constant. Herein, the functions} $\zeta _{k}$
{and} $B_{k}$ {are such as defined in} (\ref{eq3.10})
{and} (\ref{eq3.5}), {respectively, while} $P \left ( 0,T \right )$
{denotes the deterministic initial bond price}.
\end{Proposition}

\begin{proof} We substitute (\ref{eq4.2}) and (\ref{eq3.12}) into (\ref{eq4.1}) and obtain
\begin{equation*}
C_{t} = \mathbb{E}_{\mathbb{Q}} \left ( e^{I_{t} - I_{\tau }}  \bigl[ P \left (
0,T \right ) e^{G_{\tau }} -K  \bigr]^{+} | \mathcal{F}_{t} \right )
\end{equation*}
where $I_{t}$ denotes the integrated short rate process defined in (\ref{eq3.25})
and
\begin{equation*}
G_{\tau } := I_{\tau } - \sum _{k=1}^{n} \int _{0}^{\tau } \int _{D_{k}}
\zeta _{k} \left ( s,T,z \right ) d \nu _{k} \left ( z \right ) ds +
\sum _{k=1}^{n} \int _{0}^{\tau } \int _{D_{k}} \sigma _{k} B_{k} \left ( s,T
\right ) z d N_{k} \left ( s,z \right )
\end{equation*}
is a real-valued stochastic process. For $u \in \mathbb{R}$ we introduce
the deterministic function
\begin{equation*}
q \left ( u \right ) := e^{-au} \left [ P \left ( 0,T \right ) e^{u} -K
\right ]^{+}
\end{equation*}
where $a>1$ is a constant real-valued dampening parameter ensuring the
integrability of the payoff function. Indeed, it holds $q \left (
\boldsymbol{\cdot } \right ) \in \mathcal{L}^{1} \left ( \mathbb{R} \right )$.
With the latter definition at hand, we obtain
\begin{equation*}
C_{t} = \mathbb{E}_{\mathbb{Q}}  \bigl( e^{I_{t} - I_{\tau } +a G_{\tau }} q
\left ( G_{\tau } \right ) | \mathcal{F}_{t}  \bigr).
\end{equation*}With reference to \cite{8} (also see \cite{18}), we apply the inverse Fourier
transform\index{inverse Fourier transform} on the latter equation and hereafter, use Fubini's theorem which
leads us to
\begin{equation*}
C_{t} = \int _{\mathbb{R}} \hat{q}  ( y  ) \mathbb{E}_{\mathbb{Q}}
 \bigl( e^{Z_{t,\tau }} | \mathcal{F}_{t}  \bigr) dy
\end{equation*}
where we have set
\begin{equation*}
Z_{t,\tau } := I_{t} - I_{\tau } + \left ( a+iy \right ) G_{\tau }
\end{equation*}
for all $0\leq t\leq \tau $. By merging the definition of $G_{\tau }$ and
(\ref{eq3.25}) into the definition of $Z_{t,\tau }$ we deduce
\begin{equation*}
Z_{t,\tau } = I_{t} +\theta \left ( \tau ,y \right ) + \sum _{k=1}^{n}
\int _{0}^{\tau } \int _{D_{k}} \eta _{k} \left ( s,z,y \right ) d N_{k} \left (
s,z \right )
\end{equation*}
where we identified the deterministic functions $\theta \left ( \tau ,y
\right )$ and $\eta _{k} \left ( s,z,y \right )$ defined in (\ref{eq4.4}). Hence,
\begin{align*}
\mathbb{E}_{\mathbb{Q}}  \bigl( e^{Z_{t,\tau }} | \mathcal{F}_{t}  \bigr)
&= \exp \left \{  I_{t} +\theta \left ( \tau ,y \right ) + \sum _{k=1}^{n}
\int _{0}^{t} \int _{D_{k}} \eta _{k} \left ( s,z,y \right ) d N_{k} \left ( s,z
\right ) \right \}
\\
&\quad  \times \mathbb{E}_{\mathbb{Q}} \left [ \exp \left \{
\sum _{k=1}^{n} \int _{t}^{\tau } \int _{D_{k}} \eta _{k} \left ( s,z,y \right ) d
N_{k} \left ( s,z \right ) \right \}  \right ]
\end{align*}
since $I_{t}$ is $\mathcal{F}_{t}$-adapted and $\theta \left ( \tau ,y
\right )$ is deterministic. In the derivation of the latter equation, we
used the independent increment property under $\mathbb{Q}$ of the involved
pure-jump integrals. We next apply the L\'{e}vy--Khinchin formula for
additive processes (see, e.g., \cite{9,30,36}) and derive
\begin{equation*}
\mathbb{E}_{\mathbb{Q}} \left [ \exp \left \{  \sum _{k=1}^{n} \int _{t}^{\tau }
\int _{D_{k}} \eta _{k} \left ( s,z,y \right ) d N_{k} \left ( s,z \right )
\right \}  \right ] = \exp \left \{  \sum _{k=1}^{n} \overline{\psi }_{k} \left (
t,\tau ,y \right ) \right \}
\end{equation*}
where the characteristic exponents $\overline{\psi }_{k} \left ( t,\tau ,y
\right )$ are such as defined in (\ref{eq4.4}). Putting the latter equations
together, we eventually end up with (\ref{eq4.3}). The expression for the Fourier
transform $\hat{q} \left ( y \right )$ is obtained by straightforward
calculations using the definition of the function $q \left ( u \right )$.
\end{proof}

\begin{Corollary}\label{cor4.2} {In the special case} $t=0$, {the
call option price formula} (\ref{eq4.3}) {simplifies to}
\begin{equation*}
C_{0} = \int _{\mathbb{R}} \hat{q} \left ( y \right ) \exp \left \{  \theta
\left ( \tau ,y \right ) + \sum _{k=1}^{n} \overline{\psi }_{k} \left ( 0,\tau ,y
\right ) \right \}  dy
\end{equation*}
{which is deterministic}.
\end{Corollary}

%s5 #&#
\section{Practical applications}\label{sec5}

In this section, we show how the short rate model introduced in Section~\ref{sec2}
can be implemented in practical applications. For this purpose, we now
present more detailed expressions in order to prepare our model for a
possible calibration of the involved parameters. First of all, let us
recall that the increasing compound Poisson processes $L_{t}^{k}$ defined
in (\ref{eq2.3}) for every $k\in \left \{  1,\dots ,n \right \}  $ and $t\in \left [ 0,T
\right ]$ can be expressed as
\begin{equation}
\label{eq5.1}
L_{t}^{k} = \sum _{j=1}^{N_{t}^{k}} Y_{j}^{k}
\end{equation}
(cf. Section 5.3.2 in \cite{37}) where $N_{t}^{k}$ constitutes a standard
Poisson process under $\mathbb{Q}$ with deterministic jump intensity
$\alpha _{k} >0$. That is, $N_{t}^{k} \sim Poi \left ( \alpha _{k} t \right )$
such that for all $m\in \mathbb{N}_{0}$ it holds
\begin{equation*}
\mathbb{Q}  \bigl( N_{t}^{k} =m  \bigr) = \frac{\left ( \alpha _{k} t
\right )^{m}}{m!} e^{- \alpha _{k} t}.
\end{equation*} The strictly positive jump amplitudes of the L\'{e}vy process $L_{t}^{k}$
are modeled by the i.i.d. random variables $Y_{1}^{k}, Y_{2}^{k},\dots $
which take values in the set $D_{k} \subseteq \left ] 0,\infty \right [$. We
recall that the random variables $Y_{1}^{k}, Y_{2}^{k},\dots $ are
independent of the Poisson processes $N_{t}^{\overline{k}}$ for all
combinations of indices $k, \overline{k} \in \left \{  1,\dots ,n \right \}  $.
We further put $c_{k} := \mathbb{E}_{\mathbb{Q}}  [ Y_{1}^{k}
 ] \in D_{k}$ and recall that the compensated compound Poisson process
$\left ( L_{t}^{k} - c_{k} \alpha _{k} t \right )_{t\in \left [ 0,T \right ]}$
constitutes an $\left ( \mathcal{F}_{t} ,\mathbb{Q} \right )$-martingale for
each $k$ which implies
\begin{equation*}
c_{k} \alpha _{k} = \int _{D_{k}} z d \nu _{k} \left ( z \right )
\end{equation*}
due to (\ref{eq2.3}) and (\ref{eq2.4}). We stress that the Poisson processes $N_{t}^{k}$
shall not be mixed up with the Poisson random measures $d N_{k} \left ( t,z
\right )$.

In the following, we propose a number of probability distributions living
on the positive half-line (recall Section B.1.2 in \cite{37}) which constitute
suitable candidates for the modeling of the jump size distribution in our
new short rate model. As a first example, we propose to work with the
\textit{gamma distribution} and thus, assume that each random variable
$Y_{j}^{k}$ is exponentially distributed under $\mathbb{Q}$ with parameter
$\varepsilon _{k} >0$ for all $j$ and $k$. In this case, the related
L\'{e}vy measure possesses the Lebesgue density
\begin{equation}
\label{eq5.2}
d \nu _{k} \left ( z \right ) = \alpha _{k} \varepsilon _{k} e^{-
\varepsilon _{k} z} dz
\end{equation}
where $z\in D_{k} = \left ] 0,\infty \right [$ and $k\in \left \{  1,\dots ,n
\right \}  $. We find $c_{k} = {1} / {\varepsilon _{k}}$ and $Y_{j}^{k} \sim
\Gamma \left ( 1, \varepsilon _{k} \right )$. Hence, following the notation
used in Section 5.5.1 in \cite{37}, we state that we presently are in a $\Gamma
\left ( \alpha _{k}, \varepsilon _{k} \right )$-Ornstein--Uhlenbeck process
setup (also see Section 8.2 in \cite{31} and Example 15.1 in \cite{9} in this
context).

\begin{Proposition}\label{prop5.1} {Suppose that the random variables}
$Y_{j}^{k}$\ in (\ref{eq5.1}) {are exponentially distributed} (i.e.
$\Gamma \left ( 1, \varepsilon _{k} \right )$-{distributed})
{under} $\mathbb{Q}$ {with parameters} $\varepsilon _{k} >0$
{for all} $j$ {and} $k$. {Then, for} $u \in
\mathbb{R}$ {and} $t\in \left [ 0,T \right ]$ {the characteristic
function of} $L_{t}^{k}$ {is given by}
\begin{equation*}
\Phi _{L_{t}^{k}} \left ( u \right ) = \exp \left \{  \frac{iu \alpha _{k} t}{\varepsilon _{k} -iu} \right \}
\end{equation*}
{where} $\alpha _{k}$ {denotes the jump intensity of the
standard Poisson process} $N_{t}^{k}$ {appearing in} (\ref{eq5.1}).
\end{Proposition}

\begin{proof} Successively applying the definition of the characteristic
function, (\ref{eq2.3}), the L\'{e}vy--Khinchin formula and (\ref{eq5.2}), for $u
\in \mathbb{R}$ and $t\in \left [ 0,T \right ]$ we obtain
\begin{equation*}
\Phi _{L_{t}^{k}} \left ( u \right ) = \mathbb{E}_{\mathbb{Q}}  \bigl[ e^{iu
L_{t}^{k}}  \bigr] = \exp \left \{  \alpha _{k} \varepsilon _{k} t
\int _{0}^{\infty }  \bigl[ e^{iuz} -1  \bigr] e^{- \varepsilon _{k} z} dz
\right \}  .
\end{equation*} We eventually perform the integration and end up with the asserted
equality.
\end{proof}

An immediate consequence of Proposition~\ref{prop5.1} is the following representation
for the moment generating function of $L_{t}^{k}$ being valid for all $v
\in \mathbb{R}\setminus  \left \{  \varepsilon _{k} \right \}  $
\begin{equation*}
\kappa _{L_{t}^{k}} \left ( v \right ) = \Phi _{L_{t}^{k}} \left ( -iv \right ) =
\exp \left \{  \frac{v \alpha _{k} t}{\varepsilon _{k} -v} \right \}  .
\end{equation*}

\begin{Proposition}\label{prop5.2} {Assume that the random variables}
$Y_{j}^{k}$\ in (\ref{eq5.1}) {are exponentially distributed} (i.e.
$\Gamma \left ( 1, \varepsilon _{k} \right )$-{distributed})
{under} $\mathbb{Q}$ {with parameters} $\varepsilon _{k} >0$
{for all} $j$ {and} $k$. {Then, for all} $t\in \left [
0,T \right ]$, $k\in \left \{  1,\dots ,n \right \}  $ {and} $x \in
\mathbb{R}$ {the probability density function of} $L_{t}^{k}$ {under}
$\mathbb{Q}$ {takes the form}
\begin{equation*}
f_{L_{t}^{k}} \left ( x \right ) = \frac{1}{2\pi } \int _{0}^{\infty } \exp
\left \{  iu \left ( x- \frac{\alpha _{k} t}{\varepsilon _{k} +iu} \right )
\right \}  du.
\end{equation*}
\end{Proposition}

\begin{proof} First, note that it holds
\begin{equation*}
\Phi _{L_{t}^{k}} \left ( -u \right ) = \int _{0}^{\infty } e^{-iux}
f_{L_{t}^{k}} \left ( x \right ) dx =2\pi \hat{f}_{L_{t}^{k}} \left ( u
\right )
\end{equation*}
due to the definitions of the characteristic function and the Fourier
transform claimed in the sequel of (\ref{eq4.2}). We next apply the inverse Fourier
transform\index{inverse Fourier transform} which yields the density function
\begin{equation*}
f_{L_{t}^{k}} \left ( x \right ) = \frac{1}{2\pi } \int _{0}^{\infty }
\Phi _{L_{t}^{k}} \left ( -u \right ) e^{iux} du.
\end{equation*}
We finally plug the result of Proposition~\ref{prop5.1} into the latter equation
which completes the proof.
\end{proof}

We stress that Eq.~(\ref{eq5.2}) can be substituted into the corresponding formulas
appearing in the previous Propositions~\ref{prop2.2}, \ref{prop2.3}, \ref{prop3.1}, \ref{prop3.3},
\ref{prop3.7}--\ref{prop3.10} and \ref{prop4.1}
yielding more explicit expressions for the involved entities, yet
associated with gamma-distributed jump amplitudes in the underlying short
rate model. We illustrate this statement by an application of Eq.~(\ref{eq5.2}) on
Proposition~\ref{prop2.3}. The precise result reads as follows.

\begin{Proposition}\label{prop5.3} {Suppose that the random variables}
$Y_{j}^{k}\ in$ (\ref{eq5.1}) {are exponentially distributed} (i.e.
$\Gamma \left ( 1, \varepsilon _{k} \right )$-{distributed})
{under} $\mathbb{Q}$ {with parameters} $\varepsilon _{k} >0$
{for all} $j$ {and} $k$. {Let} $\sigma _{k} >0$
{be the constant volatility coefficients introduced in} (\ref{eq2.2}).
{Then, for all} $t\in \left [ 0,T \right ]$ {and} $v
\in \mathbb{R}$ {with} $v< \min _{k} \left \{  {\varepsilon _{k}} /
{\sigma _{k}} \right \}  $, $k\in \left \{  1,\dots ,n \right \}  $, {the
moment generating function under} $\mathbb{Q}$ {of the short rate
process} $r_{t}$ {reads as}
\begin{equation*}
\kappa _{r_{t}} \left ( v \right ) = \Phi _{r_{t}} \left ( -iv \right ) = \exp
\left \{  v \mu \left ( t \right ) + \sum _{k=1}^{n} \rho _{k} \left ( t,-iv
\right ) + \sum _{k=1}^{n} \psi _{k} \left ( t,-iv \right ) x_{k} \right \}
\end{equation*}
{with deterministic functions}
\begin{equation*}
\psi _{k} \left ( t,-iv \right ) =v e^{- \lambda _{k} t},\qquad  \rho _{k} \left ( t,-iv
\right ) = \frac{\alpha _{k}}{\lambda _{k}} \log \left \vert \frac{\varepsilon _{k} -v \sigma _{k} e^{- \lambda _{k} t}}{\varepsilon _{k}
-v \sigma _{k}} \right \vert .
\end{equation*}
\end{Proposition}

\begin{proof} For each $k\in \left \{  1,\dots ,n \right \}  $ we define the
deterministic functions $b_{k} \left ( s,v \right ) :=v\ \sigma _{k}
\ e^{- \lambda _{k} \left ( t-s \right )} - \varepsilon _{k}$ which satisfy
$b_{k} \left ( s,v \right ) <0$ whenever $s\in \left [ 0,t \right ]$ and $v<
\min _{k} \left \{  {\varepsilon _{k}} / {\sigma _{k}} \right \}  $. In this
setting, we combine Eq.~(\ref{eq5.2}) with the definitions of $\rho _{k}$ and
$\Lambda _{k}$ given in Proposition~\ref{prop2.3} and obtain
\begin{equation*}
\rho _{k} \left ( t,-iv \right ) =- \alpha _{k} \int _{0}^{t}
\frac{\varepsilon _{k} + b_{k} \left ( s,v \right )}{b_{k} \left ( s,v
\right )} ds.
\end{equation*}
We perform the integration and obtain the formula for $\rho _{k}$ claimed in
the proposition. The representation for the moment generating function
$\kappa _{r_{t}} \left ( v \right )$ finally follows from Proposition~\ref{prop2.3}.
\end{proof}

Other distributional choices for the random variables $Y_{j}^{k}$ modeling
the jump amplitudes would be, for example, the inverse Gaussian
distribution (see Section 5.5.2 in \cite{37}), the generalized inverse Gaussian
distribution (see Section 5.3.5 in \cite{37}) or the tempered stable
distribution (see Section 5.3.6 in \cite{37}). The related formulas for the
Lebesgue density of the L\'{e}vy measure $d \nu _{k} \left ( z \right )$
corresponding to Eq.~(\ref{eq5.2}) can be found in \cite{37}.

\begin{Remark}\label{rem5.4} We recall that the time-homogeneous compound Poisson
processes $L_{t}^{k}$ introduced in (\ref{eq2.3}) can be simulated according to
Algorithms 6.1 and 6.2 in \cite{9}. Further, in our model it is easily possible
to calculate the moments of $X_{t}^{k}$ and $r_{t}$ (see the sequel of
Proposition~\ref{prop2.2}) so that our model can be fitted to any yield curve
observed in the market by using the \textit{moment estimation method}
described in Section 7.2.2 in \cite{9}. This procedure is also called
\textit{moment matching}, as the underlying idea is to make the empirical
moments match with the theoretical moments of the model by finding a
suitable parameter set.
\end{Remark}

%s6 #&#
\section{A post-crisis model extension}\label{sec6}

In this section, we propose a post-crisis extension of the pure-jump
lower-bounded short rate model introduced in Section~\ref{sec2}. (To read more on
post-crisis interest rate models, the reader is referred to \cite{11,12,13,14,17,22,23,24,25,26,32,33,34}.)
Inspired by the modeling setups presented in \cite{33} and Chapter 2 in \cite{26},
for all $t\in \left [ 0,T \right ]$ we define the short rate spread under
$\mathbb{Q}$ by the stochastic process
\begin{equation*}
s_{t} := \mu ^{*} \left ( t \right ) + \sum _{k=n+1}^{l} X_{t}^{k}
\end{equation*}
showing a similar structure as (\ref{eq2.1}). Herein, $\mu ^{*} \left ( t \right )
\geq 0$ constitutes an integrable real-valued deterministic function and the
factors $X_{t}^{k}$ satisfy the SDE (\ref{eq2.2}), but presently for indices $k\in
\left \{  n+1,\dots ,l \right \}  $ where $l \in \mathbb{N}$ with $l>n$. Note
that it holds $s_{t} \geq \mu ^{*} \left ( t \right )\ \mathbb{Q}$-a.s. for
all $t\in \left [ 0,T \right ]$ such that the short rate spread is bounded
from below -- similar to the short rate itself [recall Remark~\ref{rem2.1} (a)]. We
interpret $s_{t}$ as an \textit{additive} spread and therefore set for all
$t\in \left [ 0,T \right ]$
\begin{equation}
\label{eq6.1}
\overline{r}_{t} := r_{t} + s_{t}
\end{equation}
(cf. \cite{12,33}) where $r_{t}$ denotes the short rate process and
$\overline{r}_{t}$ is called fictitious short rate, similarly to \cite{26}. With
reference to p. 46 in \cite{26}, we recall that the short rate spread $s_{t}$
not only incorporates credit risks, but also various other risks in the
interbank sector which affect the evolution of the LIBOR rates. Let us
moreover mention that the short rate $r_{t}$ defined in (\ref{eq2.1}) and the short
rate spread $s_{t}$ can be `correlated' by allowing for (at least) one
common factor in their respective definitions. More precisely, if the sum
in the definition of $s_{t}$ started running from $k=n$ (instead of
$k=n+1$), then the factor $X_{t}^{n}$ would appear both in the definition
of $r_{t}$ and in the definition of $s_{t}$ such that the two latter
stochastic processes would no longer be independent.

We next substitute (\ref{eq2.1}) as well as the definition of $s_{t}$ into (\ref{eq6.1})
and deduce
\begin{equation}
\label{eq6.2}
\overline{r}_{t} = \overline{\mu } \left ( t \right ) + \sum _{k=1}^{l}
X_{t}^{k}
\end{equation}
where we introduced the real-valued deterministic function $\overline{\mu }
\left ( t \right ) :=\mu \left ( t \right ) + \mu ^{*} \left ( t
\right )$. It obviously holds $\overline{r}_{t} \geq \overline{\mu } \left ( t
\right )\ \mathbb{Q}$-a.s. for all $t\in \left [ 0,T \right ]$. In accordance
to Section 3.4.1 in \cite{14}, Eq. (2.35) in \cite{26} and Section 1 in \cite{33}, we
define the fictitious bond price in our post-crisis short rate model
via
\begin{equation}
\label{eq6.3}
\overline{P} \left ( t,T \right ) := \mathbb{E}_{\mathbb{Q}} \left (
\exp \left \{  - \int _{t}^{T} \overline{r}_{u} du \right \}  \bigg|
\mathcal{F}_{t} \right )
\end{equation}
where $t\in \left [ 0,T \right ]$. The object $\overline{P} \left ( t,T
\right )$ is sometimes called artificial bond price in the literature, as it
is not physically traded in the market. Evidently, $\overline{P} \left ( T,T
\right ) =1$. Comparing (\ref{eq6.2}) with (\ref{eq2.1}) and (\ref{eq6.3}) with (\ref{eq3.3}), we see that
all our single-curve results presented in the previous sections carry over
to the currently considered post-crisis case. More precisely, the entities
$\mu \left ( t \right )$, $r_{t}$, $P \left ( t,T \right )$ and $n$ emerging in
the previous single-curve equations presently have to be replaced by
$\overline{\mu } \left ( t \right )$, $\overline{r}_{t}$, $\overline{P} \left (
t,T \right )$ and $l$, respectively. Moreover, in the present case, we
observe $\mathbb{Q}$-a.s. for all $t\in \left [ 0,T \right ]$
\begin{equation*}
0< \overline{P} \left ( t,T \right ) \leq \exp \left \{  - \int _{t}^{T}
\overline{\mu } \left ( u \right ) du \right \}
\end{equation*}
due to (\ref{eq6.3}), the lower boundedness of $\overline{r}_{t}$ and the
monotonicity of conditional expectations.

\begin{Proposition}\label{prop6.1} {It holds} $\mathbb{Q}$-a.s. for
all $t\in \left [ 0,T \right ]$
\begin{equation}
\label{eq6.4}
\overline{P} \left ( t,T \right ) \leq P \left ( t,T \right )
\end{equation}
{where} $\overline{P} \left ( t,T \right )$ {constitutes the
bond price defined in} (\ref{eq6.3}) {and} $P \left ( t,T \right )$
{is given in}~(\ref{eq3.3}).
\end{Proposition}

\begin{proof} Note that taking $\mu ^{*} \left ( t \right ) \geq 0\ \forall \ t\in \left [ 0,T \right ]$ implies $s_{t} \geq 0\ \mathbb{Q}$-a.s. for all
$t\in \left [ 0,T \right ]$. In this case, we deduce $\overline{r}_{t} \geq
r_{t}\ \mathbb{Q}$-a.s. for all $t\in \left [ 0,T \right ]$ due to (\ref{eq6.1}).
Hence, we find
\begin{equation*}
\exp \left \{  - \int _{t}^{T} \overline{r}_{u} du \right \}  \leq \exp \left \{
- \int _{t}^{T} r_{u} du \right \}
\end{equation*}
$\mathbb{Q}$-a.s. for all $t\in \left [ 0,T \right ]$. We next take
conditional expectations with respect to $\mathcal{F}_{t}$ and $\mathbb{Q}$
in the latter inequality, hereafter apply the monotonicity of conditional
expectations and finally identify (\ref{eq6.3}) and (\ref{eq3.3}) in the resulting
inequality.
\end{proof}

The result obtained in Proposition~\ref{prop6.1} possesses the following economical
interpretation: The inequality (\ref{eq6.4}) shows that nontraded bonds are
cheaper than their nonfictitious counterparts which are physically traded
in the market. This feature appears economically reasonable and stands in
accordance with the modeling assumptions and statements in \cite{12}, Section
2.1.3 in \cite{26} and Section 2.1 in \cite{33}. Moreover, combining (\ref{eq6.3}) and
(\ref{eq3.18}), we obtain (parallel to \cite{12})
\begin{equation*}
\overline{P} \left ( t,T \right ) = \exp \left \{  - \int _{t}^{T} \overline{f}
\left ( t,u \right ) du \right \}
\end{equation*}
where $\overline{f}$ is sometimes called fictitious/artificial forward rate\index{forward rate}
in the literature. It holds $\overline{f} \left ( t,t \right ) =
\overline{r}_{t}$ for all $t\in \left [ 0,T \right ]$. With reference to \cite{11}
and \cite{12}, for all $t\in \left [ 0,T \right ]$ we introduce the forward rate
spread\index{forward rate spread} via
\begin{equation*}
g \left ( t,T \right ) := \overline{f} \left ( t,T \right ) -f \left (
t,T \right )
\end{equation*}
so that we have not only set up a new pure-jump post-crisis
\textit{short} rate model, but simultaneously a new pure-jump post-crisis
\textit{forward} rate\index{forward rate} model of HJM-type in the current section. Recall that
(\ref{eq6.4}) is equivalent to $f \left ( t,u \right ) \leq \overline{f} \left ( t,u
\right )\ \mathbb{Q}$-a.s. for all $0\leq t\leq u\leq T$. From this, we
conclude that $g \left ( t,T \right ) \geq 0\ \mathbb{Q}$-a.s. for all $t\in
\left [ 0,T \right ]$. It further holds $g \left ( t,t \right ) = s_{t}$ for
all $t\in \left [ 0,T \right ]$ due to (\ref{eq6.1}).

Furthermore, in the present post-crisis setting, for a time partition
$0\leq t\leq T_{1} < T_{2}$ we define the (forward) overnight indexed swap\index{overnight indexed swap (OIS) rate}
(OIS) rate via
\begin{equation*}
F \left ( t, T_{1}, T_{2} \right ) := \frac{1}{\delta } \left (
\frac{P \left ( t, T_{1} \right )}{P \left ( t, T_{2} \right )} -1 \right )
\end{equation*}
(cf. Eq. (4.1) in \cite{26}) where $P$ denotes the zero-coupon bond price
defined in (\ref{eq3.3}) and $\delta :=\delta \left ( T_{1}, T_{2} \right )$
is the year fraction with expiry date $T_{1}$ and maturity date $T_{2}$.
With reference to \cite{12}, for $0\leq t\leq T_{1} < T_{2}$ we define the
forward LIBOR rate via
\begin{equation*}
L \left ( t, T_{1}, T_{2} \right ) := \frac{1}{\delta } \left (
\frac{\overline{P} \left ( t, T_{1} \right )}{\overline{P} \left ( t, T_{2}
\right )} -1 \right )
\end{equation*}
where $\delta $ is the year fraction and $\overline{P}$ denotes the bond
price introduced in (\ref{eq6.3}). Note that the LIBOR rate $L \left ( t, T_{1},
T_{2} \right )$ shall not be mixed up with the L\'{e}vy processes
$L_{t}^{k}$ defined in (\ref{eq2.3}). In a pre-crisis single-curve approach, it
holds $\overline{P} \left ( t,T \right ) =P \left ( t,T \right )$ which implies
$F \left ( t, T_{1}, T_{2} \right ) =L \left ( t, T_{1}, T_{2} \right )\ \mathbb{Q}$-a.s. for all $t$.

We are now prepared to derive the dynamics of the short rate spread
$s_{t}$, the fictitious short rate $\overline{r}_{t}$, the bond price
$\overline{P} \left ( t,T \right )$, the forward rate\index{forward rate} $\overline{f} \left (
t,T \right )$, the forward rate spread\index{forward rate spread} $g \left ( t,T \right )$ and the LIBOR
rate $L \left ( t, T_{1}, T_{2} \right )$. The associated computations can be
accomplished by similar techniques as presented in Sections~\ref{sec2} and \ref{sec3} and
thus, are not worked out explicitly. We
%exemplarily
provide as an example two results
without proofs in the sequel. For all $t\in \left [ 0,T \right ]$ it holds
\begin{equation*}
\frac{d \overline{P} \left ( t,T \right )}{\overline{P} \left ( t-,T \right )}
= \overline{r}_{t} dt+ \sum _{k=1}^{l} \int _{D_{k}} \zeta _{k} \left ( t,T,z
\right ) d \tilde{N}_{k}^{\mathbb{Q}} \left ( t,z \right )
\end{equation*}
where the functions $\zeta _{k}$ are such as claimed in (\ref{eq3.10}). We further
obtain in the post-crisis case
\begin{align*}
L \left ( t, T_{1}, T_{2} \right ) &= \frac{1}{\delta } \Biggl(
\frac{\overline{P} \left ( 0, T_{1} \right )}{\overline{P} \left ( 0, T_{2}
\right )} \times \prod _{k=1}^{l} \exp  \Biggl\{  \int _{0}^{t} \int _{D_{k}}
\Psi _{k} \left ( s,z \right ) d \nu _{k} \left ( z \right ) ds
\\
&\quad  + \int _{0}^{t}
\int _{D_{k}} \Sigma _{k} \left ( s,z \right ) d N_{k} \left ( s,z \right )
\Biggr\}  -1  \Biggr)
\end{align*}
with deterministic functions
\begin{align*}
\Psi _{k} \left ( s,z \right ) &:= e^{\sigma _{k} B_{k} \left ( s, T_{2}
\right ) z} - e^{\sigma _{k} B_{k} \left ( s, T_{1} \right ) z} <0,
\\
 \Sigma _{k}
\left ( s,z \right ) &:= \sigma _{k} \left [ B_{k} \left ( s, T_{1}
\right ) - B_{k} \left ( s, T_{2} \right ) \right ] z>0.
\end{align*}

%\begin{appendix}
%\end{appendix}

%\begin{acknowledgement}%[title={Acknowledgments}]
%\end{acknowledgement}

%\begin{funding}
%\gsponsor[id=,sponsor-id=]{}
%\gnumber[refid=]{}
%\end{funding}

% structpyb loaded by romualda, 2020-04-08 14:22:04

\end{document}